\newcommand{\longversion}[1]{#1}
\newcommand{\shortversion}[1]{}

\documentclass[letterpaper]{article} 
\usepackage{aaai20}  
\usepackage{times}  
\usepackage{helvet} 
\usepackage{courier}  
\usepackage[hyphens]{url}  
\usepackage{graphicx} 
\def\UrlFont{\rm}  
\usepackage{graphicx}  
\usepackage{algorithm2e} 
\usepackage{subfloat} 
\usepackage{subfig} 
\usepackage{balance}

\usepackage{amssymb,amsmath,amsthm}
\usepackage[usenames,dvipsnames]{color}
\usepackage{makecell}
\usepackage{multirow}
\usepackage{graphicx}

\usepackage[usenames,svgnames,dvipsnames]{xcolor}

\longversion{
\usepackage{fullpage}
}

\newdimen\prevdp
\def\leftlabel#1{\noalign{\prevdp=\prevdepth
   \kern-\prevdp\nointerlineskip\vbox to0pt{\vss\hbox{\ensuremath{#1}}}\kern\prevdp}}

\usepackage{url}

\usepackage{enumerate}

\usepackage{xspace}
\newcommand{\NP}{\ensuremath{\mathsf{NP}}\xspace}
\newcommand{\NPC}{\ensuremath{\mathsf{NP}}-complete\xspace}
\newcommand{\el}{\ensuremath{\ell}\xspace}
\newcommand{\suc}{\ensuremath{\succ}\xspace}

\let\mydelta\delta
\renewcommand{\delta}{\ensuremath{\mydelta}\xspace}
\let\myalpha\alpha
\renewcommand{\alpha}{\ensuremath{\myalpha}\xspace}
\let\mystar\star
\renewcommand{\star}{\ensuremath{\mystar}}

\newcommand{\YES}{{\sc yes}\xspace}

\newcommand{\SAT}{{\sc SAT}\xspace}
\newcommand{\FP}{{\sc Fair Partitioning}\xspace}
\newcommand{\FD}{{\sc Fair Districting}\xspace}
\newcommand{\FCD}{{\sc Fair Connected Districting}\xspace}

\newcommand{\DCP}{{\sc \ensuremath{2}-Disjoint Connected Partitioning}\xspace}
\newcommand{\GP}{{\sc Greedy Partitioning}\xspace}
\newcommand{\GD}{{\sc Greedy Districting}\xspace}
\newcommand{\GCD}{{\sc Greedy Connected Districting}\xspace}

\renewcommand{\AA}{\ensuremath{\mathcal A}\xspace}

\newcommand{\CC}{\ensuremath{\mathcal C}\xspace}
\newcommand{\DD}{\ensuremath{\mathcal D}\xspace}
\newcommand{\EE}{\ensuremath{\mathcal E}\xspace}

\newcommand{\GG}{\ensuremath{\mathcal G}\xspace}
\newcommand{\HH}{\ensuremath{\mathcal H}\xspace}

\newcommand{\LL}{\ensuremath{\mathcal L}\xspace}

\newcommand{\OO}{\ensuremath{\mathcal O}\xspace}

\newcommand{\TT}{\ensuremath{\mathcal T}\xspace}
\newcommand{\UU}{\ensuremath{\mathcal U}\xspace}
\newcommand{\VV}{\ensuremath{\mathcal V}\xspace}

\newcommand{\XX}{\ensuremath{\mathcal X}\xspace}
\newcommand{\YY}{\ensuremath{\mathcal Y}\xspace}
\newcommand{\ZZ}{\ensuremath{\mathcal Z}\xspace}

\newcommand{\RB}{\ensuremath{\mathbb R^+}\xspace}

\usepackage{nicefrac}

\newtheorem{proposition}{\bf Proposition}
\newtheorem{observation}{\bf Observation}
\newtheorem{theorem}{\bf Theorem}

\newtheorem{definition}{\bf Definition}

\newcommand{\eps}{\ensuremath{\varepsilon}\xspace}
\renewcommand{\epsilon}{\eps}

\newcommand{\ignore}[1]{}

\newcommand{\pr}{\ensuremath{\prime}}
\newcommand{\prr}{\ensuremath{{\prime\prime}}}

\renewcommand{\ge}{\geqslant}
\renewcommand{\le}{\leqslant}

\usepackage{cleveref}

\crefname{theorem}{theorem}{\bf Theorem}
\crefname{observation}{observation}{\bf Observation}
\crefname{lemma}{lemma}{\bf Lemma}
\crefname{corollary}{corollary}{\bf Corollary}
\crefname{proposition}{proposition}{\bf Proposition}
\crefname{definition}{definition}{\bf Definition}
\crefname{claim}{claim}{\bf Claim}
\crefname{reductionrule}{reduction rule}{\bf Reduction rule}

\title{Minimizing Margin of Victory for Fair Political and Educational Districting}

\author{
Ana-Andreea Stoica,\textsuperscript{\rm 1,2}
Abhijnan Chakraborty,\textsuperscript{\rm 1}
Palash Dey,\textsuperscript{\rm 3}
Krishna P. Gummadi\textsuperscript{\rm 1} \\
\textsuperscript{1}Max Planck Institute for Software Systems, Germany\\
\textsuperscript{2}Columbia University, USA\\
\textsuperscript{3}Indian Institute of Technology Kharagpur, India
}

\begin{document}

\maketitle

\begin{abstract}
In many practical scenarios, a population is divided into disjoint groups for better administration, e.g., electorates into political districts, employees into departments, students into school districts, and so on. However, grouping people arbitrarily may lead to biased partitions, raising  concerns of gerrymandering in political districting, racial segregation in schools, etc. To counter such issues, in this paper, we conceptualize such problems in a voting scenario, and propose \FD problem to divide a given set of people having preference over candidates into $k$ groups such that the {\it maximum margin of victory} of any group is minimized. We also propose the \FCD problem which additionally requires each group to be connected. We show that the \FD problem is \NPC for plurality voting even if we have only $3$ candidates but admits polynomial time algorithms if we assume $k$ to be some constant or everyone can be moved to any group. In contrast, we show that the \FCD problem is \NPC for plurality voting even if we have only $2$ candidates and $k = 2$. Finally, we propose heuristic algorithms for both the problems and show their effectiveness in UK political districting and in lowering racial segregation in public schools in the US.
\end{abstract}


\section{Introduction}
Dividing a population into smaller groups is often a practical necessity for better administration. For example, in many democratic countries (most notably, the countries who follow the Westminster System e.g., UK, Canada, India, Australia, or the Presidential System e.g., US, Brazil, Mexico, Indonesia), electorates are divided into electoral districts; in many organizations, employees are divided into administrative units like departments; students enrolled in public schools in the US are divided into school districts; and so on. However, the population is not homogeneous, it consists of people with different attributes -- gender, race, religion, ideology, etc. Dividing people arbitrarily may lead to biased grouping, manifested differently in different contexts.

In electoral districting, given the voting pattern of the electorate, political parties in power may draw the district boundaries that favor them -- a practice termed as {\it gerrymandering}~\cite{lewenberg2017divide}. For example, majority of the opposition supporters may be assigned to a few districts, such that the opponents become minority in other districts. Alternatively, the ruling party may want to ensure that it enjoys a healthy lead over its opponents in many districts, so that if a handful of its supporters change sides, it does not hamper the winnability. There have been several instances of such manipulations in electoral (re)districting in the US, starting as early as in 1812, by then Massachusetts governor Elbridge Gerry (the term {\it gerrymandering} originated after him)~\cite{issacharoff2002gerrymandering}. 

Public schools in the US are governed by school boards representing local communities, and are largely funded from local property taxes~\cite{corsi2015guide,chakraborty2019nudging}. Most of the students go to a school in the district they live, proximity playing an important role in the school choice~\cite{urban2015closest}. Thus, the way schools are distributed determines racial composition of their students and the revenue they earn. Several reports claim that wealthier, whiter communities have pushed policies so that white families can live in white-majority areas and attend white-majority schools~\cite{richards2014gerrymandering,chang2018we}. Despite the desegregating efforts following the landmark Supreme Court verdict in {\it Brown v. Board of Education case} in 1954 (which ruled racial segregation of children in public schools to be unconstitutional), $63\%$ of classmates of a white student are whites, compared to $48\%$ of all students being whites; similarly, $40\%$ of black and Hispanic students attend schools where over $90\%$ students are people of color~\cite{franken2019what}. As a consequence, a recent report by an educational non-profit EdBuild claimed that {\it ``Non-white school districts get \$23 Billion less than white districts, despite serving the same number of students"}~\cite{EdBuild}.

To counter such unfairness issues while dividing people into groups, in this paper, we conceptualize such problems in a voting scenario: the goal is to divide a set of $n$ people, each having a preference over a set of candidates, into $k$ groups. While the mapping of electoral districting into voting is direct and utilizes people's ideological preferences, we can think of context-specific mapping in other scenarios. For example, in school districting, we can think of students having preference according to their sensitive attributes (e.g., gender, race, etc.). 
Once the mapping is done, we propose the \FD problem to create $k$ groups such that the maximum {\it margin of victory} of any group is minimized, where {\it margin of victory} is defined as the number of people who need to change their preference to change the winner. We also propose the \FCD problem which additionally requires each group to be connected. Reducing margin of victory would lead to everyone's opinion within a group to be valuable, since the consensus of the group 
can be changed even if a small number of people change their preferences. In practical applications, it would lead to higher accountability from the elected candidate in political districting, lower racial segregation in schools, increase inter-discipline exchange in academic departments,  
and so on. 


\vspace{-2mm}
\subsection{Contribution}
We make the following contributions in this paper.
\begin{itemize}
\item We show that the \FD problem is \NPC for the plurality voting rule even when we have only $3$ alternatives and there is no constraint on the size of individual groups~[\Cref{thm:fp_npc}]. We complement this intractability result by proving existence of polynomial time algorithms,  
when (i) every voter can be moved to any group (which we term as the \FP problem)~[\Cref{thm:fp_poly_alt}], and (ii) we have a constant number of groups~[\Cref{thm:fp_poly_grp}] for the plurality voting rule.
  
\item We show that the \FCD problem is \NPC for the plurality voting rule even when there is only $2$ alternatives, $2$ districts, the maximum degree of any vertex in the underlying graph is $5$, and no constraint on the size of districts~[\Cref{thm:fd_npc}]. 
This shows that, although both \FD and \FCD problems are \NPC, \FCD problem is computationally harder than the \FD problem.
 
 \item We propose heuristic algorithms for both \FD and \FCD problems and show their effectiveness 
 in reducing margin of victory in electoral districts in the UK, as well as in lowering racial segregation in public schools in the US. 
\end{itemize}

\vspace{-2mm}
\subsection{Related Work}
Voting mechanisms have been at the center of historical, political, and sociological studies~\cite{barbara1987reliability,barbera1991voting,lublin1999paradox,erdelyi2015more}, due to their impact on local communities and society at large. The problem of \textit{unfair distribution of voters into districts}, otherwise known as {\it gerrymandering}, has received significant attention from academic researchers~\cite{butler1992redrawing,johnston2006disproportionality,issacharoff2002gerrymandering,bachrach2016misrepresentation}, and in particular from the computational social choice theorists, setting geographical~\cite{lewenberg2017divide} and social constraints~\cite{cohen2018gerrymandering,ito2019algorithms,borodin2018big} to population mobility. 

Central to the problem of gerrymandering is the concept of \textit{representation}: does a collective represent the choices or attributes of those comprising it? In other words, does a district represent its voters? While recent papers conceptualize different measures of representation in voting scenarios~\cite{bachrach2016misrepresentation,johnston2006disproportionality,feix2008majority,gelman2002mathematics,felsenthal2015election}, to our knowledge, 
we are the first to use the concept of margin of victory for re-districting voters 
to achieve better representation. While minimizing margin of victory does not ensure proportional representation of all voter choices in each district, it at least ensures that the voices present are not lost in the crowd. Intuitively, lowering margin of victory across districts would ensure a strong opposition to each majority winner, safeguarding against district monopolies, as well as against diluting voter power across many districts. 

Computing the margin of victory for different voting rules has been studied in~\cite{xia2012computing}, while~\cite{dey2015estimating} and~\cite{blom2018computing} estimate it in real elections. However, to our knowledge, the problem of minimizing margin of victory has not 
attracted much attention. In this paper, we characterize the complexity of this problem for plurality voting, one of the most common voting rules, and give practical algorithms to solve it in real and synthetic datasets. 



\section{Preliminaries}\label{sec:prelim}

\subsection{Voting Setting}

For a positive integer $k$, we denote the set $\{1, 2, \ldots, k\}$ by $[k]$. Let $\AA=\{a_i: i\in[m]\}$ be a set of $m$ {\em alternatives}. A complete order over the set \AA of alternatives is called a {\em preference}. We denote the set of all possible preferences over \AA by $\LL(\AA)$. A tuple $(\suc_i)_{i\in[n]}\in\LL(\AA)^n$ of $n$ preferences is called a {\em profile}. An {\em election} \EE is a tuple $(\suc,\AA)$ where \suc is a profile over a set \AA of alternatives. If not mentioned otherwise, we denote the number of alternatives and the number of preferences by $m$ and $n$ respectively. A map $r:\uplus_{n,|\AA|\in\mathbb{N}^+}\LL(\AA)^n\longrightarrow 2^\AA\setminus\{\emptyset\}$ is called a \emph{voting rule}. In this paper we consider the {\em plurality voting rule} where the set of winners is the set of alternatives who appear at the first position of a highest number of alternatives. We say that a voter votes for an alternative if the voter prefers that alternative most. The number of preferences where an alternative appears at the first place is called her plurality score.

\subsection{Margin of Victory}

Let $r$ be any voting rule. The margin of victory of an election $((\suc_i)_{i\in[n]}, \AA)$ is the minimum number of votes that needs to be changed to change the election outcome. It easily follows that the margin of victory of a plurality election is the ceiling of half the difference between the plurality score of the two highest plurality scores of the alternatives. For ease of notation, we assume that the margin of victory of an empty election (no voters) is $\infty$.

\subsection{Problem Definition}

We now define our basic problem which we call \FD.

\begin{definition}[\FD]\label{def:fd}
Given a set \AA of $m$ alternatives, a set \VV of $n$ voters along with their corresponding preferences, a set of $k$ groups $\HH=\{H_i, i\in[k]\}$ along with the set $\VV_i$ of voters corresponding to each group $H_i$ for $i\in[k]$ such that $(\VV_i)_{i\in[k]}$ forms a partition of \VV, a function $\pi:\VV\longrightarrow 2^\HH\setminus\{\emptyset\}$ denoting the set of groups where each voter can be part of, minimum size $s_{min}$ and maximum size $s_{max}$ of every group, and a target $t$ of maximum margin of victory of any group, compute if there exists a partition $(\VV_i^\pr)_{i\in[k]}$ of \VV into these $k$ groups such that the following holds.
\begin{enumerate}[(i)]
 \item For every $i\in[k]$ and $v\in\VV_i^\pr$, we have $H_i\in\pi(v)$
 \item For every $i\in[k]$, we have $s_{min}\le|\VV_i^\pr|\le s_{max}$
 \item The margin of victory in the group $H_i$ is at most $t$ for every $i\in[k]$
\end{enumerate}
We denote an arbitrary instance of this problem by $(\AA,\VV,k,\HH=(H_i)_{i\in[k]},(\VV_i)_{i\in[k]},\pi,s_{min},s_{max},t)$.
\end{definition}

An important special case of \FD is when every voter can be moved to every district; that is $\pi(v)=\HH$ for every voter $v\in\VV$. We call this problem \FP. We denote an arbitrary instance of \FP by $(\AA,\VV,k,\HH=(H_i)_{i\in[k]},(\VV_i)_{i\in[k]},s_{min},s_{max},t)$.

The \FD problem is generalized to define the \FCD problem where the input also have a graph defined on the set of voters, the given districts are all connected, and we require the new districts to be connected as well. We denote an arbitrary instance of \FCD by $(\AA,\VV,\GG,k,\HH=\left(H_i\right)_{i\in[k]},\left(\VV_i\right)_{i\in[k]},\pi,s_{min},s_{max},t)$. In this paper, we study the above problems only for the plurality voting rule and thus omit specifying it every time.

\vspace{1mm}
\noindent The following observation is immediate from problem definitions itself.

\begin{observation}\label{prop:fp_fd}
	\FP many to one reduces to \FD which many to one reduces to \FCD.
\end{observation}

\section{Results: Intractability}

In this section, we present our hardness results. Our first result shows that \FD is \NPC even with $3$ alternatives. For that we reduce from the well known \SAT problem which is known to be \NPC.

\begin{theorem}\label{thm:fp_npc}
 The \FD problem is \NPC even if we have only $3$ alternatives and there is no constraint on the size of any district.
\end{theorem}

\begin{proof}
 \FD clearly belongs to \NP. To prove \NP-hardness, we reduce from the \SAT problem. Let $\left(\XX=\left\{x_i: i\in[n]\right\}, \CC=\left\{C_j: j\in[m]\right\}\right)$ be an arbitrary instance of \SAT. Let us consider the following instance $(\AA,\VV,k,\HH=,(\VV_i)_{i\in[k]},\pi,s_{min}=0,s_{max}=\infty,t=2)$ of \FD.
 \begin{align*}
  \AA &= \{a,b,c\}\\
  \HH &= \{ \XX_i, \bar{\XX_i}, \ZZ_i: i\in[n] \} \cup \{\YY_j: j\in[m]\}\\
  \forall i\in[n],& \text{ Votes in } \XX_i : m+2 \text{ votes for } a\\
  & m \text{ votes for } b, m-1 \text{ votes for } c\\
  \forall i\in[n],& \text{ Votes in } \bar{\XX_i} : m+2 \text{ votes for } a\\
  & m \text{ votes for } b, m-1 \text{ votes for } c\\
  \forall i\in[n],& \text{ Votes in } \ZZ_i : m+2 \text{ votes for } a\\
  & m+1 \text{ votes for } c\\
  \forall j\in[m],& \text{ Votes in } \YY_j : m+3 \text{ votes for } a\\
  & m \text{ votes for } b\\
  t &= 2
 \end{align*}
 
 Let $f$ be a function defined on the set of literals as $f(x_i)=\XX_i$ and $f(\bar{x}_i)=\bar{\XX_i}$ for every $i\in[n]$. We now describe the $\pi$ function. For $i\in[n]$, no voter in $\ZZ_i$ can move to any other district except one voter who votes for the alternative $c$ and she can move to $\XX_i$ and $\bar{\XX_i}$. For $i\in[n]$, no voter voting for the alternatives $a$ or $c$ in both $\XX_i$ and $\bar{\XX_i}$ leave their current districts; a voter in $\XX_i$ ($\bar{\XX_i}$ respectively) who votes for the alternative $b$ can move to the district $\YY_j$ for some $j\in[m]$ if the variable $x_i$ ($\bar{x_i}$ respectively) appears in the clause $C_j$. Finally no voter in the district $\YY_j, j\in[m]$ leave their current district. This finishes the description of $\pi$ and the description of the instance of \FD. We claim that the two instances are equivalent.
 
 In one direction, let us assume that the \SAT instance is a \YES instance. Let $g:\XX\longrightarrow\{0,1\}$ be a satisfying assignment for the \SAT instance. Let us consider the following movement of the voters -- for $i\in[n],$ if $g(x_i)=1$, then one voter in the district $\ZZ_i$ which votes for the alternative $c$ moves to the district $\XX_i$; otherwise one voter in the district $\ZZ_i$ which votes for the alternative $c$ moves to the district $\bar{\XX_i}$. For $j\in[m]$, let the clause $C_j$ be $ \el_1 \vee \el_2 \vee \el_3$ and $g$ sets the literal $\el_1$ to be $1$ (we can assume this without loss of generality). Then one voter from the district $f(\el_1)$ who votes for $b$ moves to the district $\YY_j$. Since the assignment $g$ satisfies all the clauses, the margin of victory in the district $\YY_j$ is $2$ for every $j\in[m]$. For $i\in[n]$, if $g(x_i)=0$ ($g(x_i)=1$ respectively), then the margin of victory in the district $\bar{\XX_i}$ ($\XX_i$ respectively) is $2$ since it receives a voter voting for the alternative $c$. The rest of the districts (for $i\in[n], \XX_i$ if $g(x_i)=0$ and $\bar{\XX_i}$ if $g(x_1)=1$) remain same and their margin of victory remains to be $2$. Hence the \FP instance is also a \YES instance.
 
 In the other direction, let us assume that the \FP instance is a \YES instance. We define an assignment $g:\XX\longrightarrow\{0,1\}$ to the variables in the \SAT instance as follows. For $i\in[n]$, if a voter in the district $\ZZ_i$ who votes for $c$ moves to $\XX_i$, then we define $g(x_i)=1$; otherwise we define $g(x_i)=0$. We claim that $g$ is a satisfying assignment for the \SAT instance. Suppose not, then there exists a clause $C_j = \el_1 \vee \el_2 \vee \el_3$ for some $j\in[m]$ which $g$ does not satisfy. To make the margin of victory of the district $\YY_j$ at most $2$, one voter who votes for $b$ must move into $\YY_j$ either from the district $f(\el_1)$ or from the district $f(\el_2)$ or from the district $f(\el_3)$. However, since $g$ does not set any of $\el_1, \el_2,$ or $\el_3$ to $1$, none of these districts receive any voter who votes for the alternative $c$. Consequently, none of the district can send a voter who votes for the alternative $b$ to the district $\YY_j$ since otherwise the margin of victory of district which sends a voter who votes for the alternative $b$ becomes at least $3$ contradicting our assumption that the \FP instance is a \YES instance. Hence the \SAT instance is a \YES instance.
\end{proof}

Due to 
Observation 1, it follows immediately from Theorem 1 that the \FCD problem for plurality voting rule is also \NPC. We next show that \FCD is \NPC even if we simultaneously have $2$ alternatives and $2$ districts. For that, we reduce from \DCP \longversion{which is defined as follows.}\shortversion{which is known to be \NPC~\cite{DBLP:journals/tcs/HofPW09}. We omit the proof of \Cref{thm:fd_npc} due to space restriction. It can be found in the supplementary material.}
\longversion{
\begin{definition}[\DCP]
 Given a connected graph $\GG=(\VV,\EE)$ and two disjoint nonempty sets $\ZZ_1,\ZZ_2\subset\VV$, compute if there exists a partition $(\VV_1,\VV_2)$ of \VV such that $\ZZ_1\subseteq\VV_1, \ZZ_2\subseteq\VV_2, \GG[\VV_1]$ and $\GG[\VV_2]$ are both connected. We denote an arbitrary instance of \DCP by $(\GG,\ZZ_1,\ZZ_2)$.
\end{definition}

It is already known that the \DCP problem is \NPC~\cite[Theorem 1]{DBLP:journals/tcs/HofPW09}. However the proof of Theorem 1 in \cite{DBLP:journals/tcs/HofPW09} can be imitated as a reduction from the version of \SAT where every literal appears in exactly two clauses; this version is also known to be \NPC~\cite{ECCC-TR03-022}. This proves the following.

\begin{proposition}
 The \DCP problem is \NPC even if the maximum degree of the input graph is $5$.
\end{proposition}

}

\begin{theorem}\label{thm:fd_npc}\shortversion{[$\star$]}
 The \FCD problem is \NPC even if we have only $2$ alternatives, $2$ districts, the maximum degree of any vertex in the underlying graph is $5$, and we do not have any constraint on the size of districts.
\end{theorem}

\longversion{
\begin{proof}
 The \FCD problem is clearly in \NP. To prove \NP-hardness, we reduce from \DCP to \FCD. Let $(\GG^\pr=(\UU,\EE^\pr),\ZZ_1,\ZZ_2)$ be an arbitrary instance of \FCD. Without loss of generality, let us assume that the degree of every vertex in $\ZZ_2$ is $2$; $z_2$ be any arbitrary (fixed) vertex of $\ZZ_2$. Let us consider the following instance $(\AA,\VV,\GG=(\VV,\EE),k=2,\HH=(H_i)_{i\in[2]},(\VV_i)_{i\in[2]},\pi,s_{min}=0,s_{max}=\infty,t=1)$ of \FCD.
 \begin{align*}
  \AA &= \{x,y\}\\
  \VV &= \{v_z: z\in\ZZ_2\} \\
  &\cup \{v_u, w_u: u\in\VV\setminus\ZZ_2\}\\
        &\cup\DD, \DD=\left\{d_i: i\in\left[\left|\ZZ_2\right|+1\right]\right\}\\
  \EE &= \{\{v_a,v_b\}: \{a,b\}\in\EE^\pr\}\\
  	& \cup \{\{v_u,w_u\}:u\in\VV[\GG^\pr]\setminus\ZZ_2\}\\
        &\cup\{\{d_i,d_j\}:i,j\in\left[\left|\ZZ_2\right|+1\right],j=i+1\}\\
       & \cup \{\{z_2,d_1\}\}\\
  \HH_2 &= \left\{d_i: i\in\left[\left|\ZZ_2\right|+1\right]\right\}\\
  \HH_1 &= \VV\setminus\HH_2\\
  &\text{Vote of } v_u, u\in\VV: x\suc y\\
  &\text{Vote of } w_u, u\in\VV\setminus\ZZ_2: y\suc x\\
  &\text{Vote of } d_i, i\in\left[\left|\ZZ_2\right|+1\right]: y\suc x\\
  &\pi(v_z) = \{\HH_1\}, z\in\ZZ_1\\
  &\pi(d_i) = \{\HH_2\}, i\in\left[\left|\ZZ_2\right|+1\right]\\
  &\pi(v) = \{\HH_1, \HH_2\} \text{ for every other vertex } v
 \end{align*}
 
 This finishes the description of the instance of \FCD. We now claim that the \FCD instance is equivalent to the \DCP instance.
 
 In one direction, let us assume that the \DCP instance is a \YES instance. Let $(\VV_1,\VV_2)$ be a partition of \UU such that $\ZZ_1\subseteq\VV_1, \ZZ_2\subseteq\VV_2, \GG^\pr[\VV_1]$ and $\GG^\pr[\VV_2]$ are both connected. We consider the following new partition of the voters.
 \begin{align*}
  \text{Voters of } \HH_1: \{v_u, w_u: u\in\VV_1\}; \text{ voters of } \HH_2: \text{ others}
 \end{align*}
 
 Since $\GG^\pr[\VV_1]$ is connected, it follows that $\GG[\HH_1]$ is also connected. Similarly, since $\GG^\pr[\VV_2]$ is connected, $\GG[\DD]$ is connected, and $\{z_2,d_1\}\in\EE[\GG]$, it follows that $\GG[\HH_2]$ is also connected. In $\HH_1$, both the alternatives $x$ and $y$ receive the same number of votes and thus the margin of victory of $\HH_1$ is $1$. In $\HH_2$, the alternatives $x$ receives $1$ less vote than the alternatives $y$ and thus the margin of victory of $\HH_2$ is $1$. Thus the \FCD instance is also a \YES instance.
 
 In the other direction, let us assume that there exists a valid partition $(\HH_1^\pr,\HH_2^\pr)$ of the voters such that both $\GG[\HH_1^\pr]$ and $\GG[\HH_2^\pr]$ are connected and the margin of victory of both $\HH_1^\pr$ and $\HH_2^\pr$ are $1$. Let us define $\VV_1=\{u\in\VV[\GG^\pr]: v_u\in\HH_1^\pr\}$ and $\VV_2=\VV[\GG^\pr]\setminus\VV_1$. It follows from the function $\pi$ that we have $\ZZ_1\subseteq\VV_1^\pr, \ZZ_2\subseteq\VV_2^\pr$. Also $\GG^\pr[\VV_1^\pr]$ is connected since the voters in $\HH_1^\pr$ are connected. We also have $\GG^\pr[\VV_2^\pr]$ is connected since the voters in $\HH_2^\pr$ are connected, the vertices in \DD forms a path, and there exists a pendant vertex in \DD. We also have $\ZZ_2\in\VV_2^\prr$ since the voters in $\{v_u:\in\ZZ_2\}$ belongs to $\HH_2$; otherwise the margin of victory of $\HH_2$ would be strictly more than $1$. Hence $(\VV_1^\pr, \VV_2^\pr)$ is a solution of the \DCP instance and thus the instance is a \YES instance.
\end{proof}
}

\section{Results: Polynomial Time Algorithms}

We now present out polynomial time algorithms. We first show that \FP is polynomial time solvable.

\begin{theorem}\label{thm:fp_poly_alt}
 The \FP problem is polynomial time solvable if the number of alternatives is a constant.
\end{theorem}

\begin{proof}
 An arbitrary instance of \FP be $(\AA,\VV,k,\HH=\left(H_i\right)_{i\in[k]},\left(\VV_i\right)_{i\in[k]},s_{min},s_{max},t)$. For an alternative $a\in\AA$, let $n_a$ be the number of vote that $a$ receives. We present a dynamic programming based algorithm for the \FP problem. The dynamic programming table $\TT\left(\left(i_a\in\{0,1,\ldots,n_a\}\right)_{a\in\AA},\el\in[k]\right)$ is defined as follows -- $\TT\left(\left(i_a\right)_{a\in\AA},\el\right)$ is the minimum integer $\lambda$ such that the voting profile consisting $i_a$ number of voters voting for the alternative $a$ can be partitioned into $\el$ districts such that the margin of victory of any district is at most $\lambda$. For every $i_a\in\{0,1,\ldots,n_a\}, a\in\AA$, we initialize $\TT\left(\left(i_a\right)_{a\in\AA},1\right)$ to be the margin of victory of the voting profile which consists of $i_a$ number of voters voting for the alternative $a$ for $a\in\AA$. We update the entries in the table \TT as follows for every $\el\in\{2,3,\ldots,k\}$.
 \begin{align*}
 	&\TT\left(\left(i_a\right)_{a\in\AA},\el\right)\\
 	&= \min_{\substack{(i_a^\pr)_{a\in\AA}, i_a^\pr\ge 0\;\forall a\in\AA\\s_{min}\le\sum_{a\in\AA} i_a^\pr\le s_{max}}} \max 
 	\left\{\!\begin{aligned}
 	          & mv\left(\left(i_a^\pr\right)_{a\in\AA}\right),\\
 	          & \TT\left(\left(i_a - i_a^\pr\right)_{a\in\AA}, \el-1\right)
 	         \end{aligned}\right\}
 \end{align*}
 
 In the above expression $mv\left(\left(i_a^\pr\right)_{a\in\AA}\right)$ denotes the plurality margin of victory of the profile which consists of $i_a^\pr$ number of voters voting for the alternative $a$ for $a\in\AA$. Updating each entry of the table takes $\OO\left(\prod_{a\in\AA}n_a\right)\text{poly}(m,n)$ time. The table has $k\prod_{a\in\AA}n_a$ entries. Hence the running time of our algorithm is $\OO\left(\prod_{a\in\AA}n_a^2\right)\text{poly}(m,n) = \OO\left(n^{2m}\text{poly}(m,n)\right)$ which is $n^{\OO(1)}$ when we have $m=\OO(1)$.
\end{proof}

We next present a polynomial time algorithm for \FD if we have a constant number of districts.

\begin{theorem}\label{thm:fp_poly_grp}
 The \FD problem is polynomial time solvable if the number of districts is a constant.
\end{theorem}

\begin{proof}
 An arbitrary instance of \FD be $(\AA,\VV,k,\HH=\left(H_i\right)_{i\in[k]},\left(\VV_i\right)_{i\in[k]},\pi,s_{min},s_{max},t)$. We guess a winner and a runner up of every district -- there are ${m\choose 2}^k=\OO(m^{2k})$ possibilities. We also guess the plurality score of a winner of every district -- there are $\OO(n^k)$ possibilities. Given a guess of a winner, its plurality score, and a runner up alternative of every district, we reduce the problem of computing if there exists a partition of \VV (respecting the given guesses) which achieves the maximum margin of victory of at most $t$ to a $s^\pr$ to $t^\pr$ flow problem (with demand on edges) instance $\left(\GG=\left(\UU,\EE\right),c:\EE\longrightarrow\RB, d:\EE\longrightarrow\RB\right)$ as follows.
 \begin{align*}
  \UU &= \UU_L \cup \UU_M \cup \UU_R \cup \{s^\pr,t^\pr\} \text{ where }\\
  \UU_L &= \{u_v: v\in\VV\}\\
  \UU_M &= \{u_{a,i}: a\in\AA,i\in[k]\}\\
  \UU_R &= \{u_i: i\in[k]\}\\
  \EE &= \left\{(s^\pr, u_v): v\in\VV\right\}\\
  & \cup \left\{(u_v, u_{a,i}): v\in\VV,i\in[k],\right.\\
  &\left. v \text{'s vote is } a\suc\cdots, \HH_i\in\pi(v)\right\}\\
  & \cup \left\{(u_{a,i},u_i):a\in\AA,i\in[k]\right\}\\
  & \cup \left\{(u_i,t^\pr):i\in[k]\right\}\\
 \end{align*}
 The capacity $c$ of every edge from $s^\pr$ to $\UU_L$ and from $\UU_L$ to $\UU_M$ is $1$. For every $i\in[k]$, if $x$ and $y$ are respectively the (guessed) winner and runner up of $\HH_i$ and $n_i$ is the (guessed) plurality score of a winner in $\HH_i$, then we define the capacity and demand of the edge $(u_{x,i},u_i)$ to be $n_i$ and the capacity and demand of the edge $(u_{y,i},u_i)$ to be $(n_i-t)$; if $(n_i-t)$ is not positive, then we discard the current guess. We define the capacity of the edge $(u_{z,i},u_i)$ to be $n_i$ for every alternative $z$ who is not the guessed winner in $\HH_i$ for $i\in[k]$. Finally we define the capacity and demand of every edge from $U_R$ to $t^\pr$ to be $s_{max}$ and $s_{min}$ respectively. We claim that the given \FD instance is a \YES instance if and only if there exists a guess of a winner, its plurality score, and a runner up alternative of every district whose corresponding flow instance has an $s^\pr$ to $t^\pr$ flow of value $n$.
 
 In one direction, suppose the \FD instance is a \YES instance. Let $x_i$ and $y_i$ be a winner and a runner up respectively in $\HH_i$ and $n_i$ be the plurality score of a winner in $\HH_i$ for $i\in[k]$. For the guess corresponding to the solution of \FD, we send $1$ unit of flow from $s^\pr$ to $u_v, v\in\VV$, from $u_v$ to $u_{a,i}$ if the voter $v$ belongs to $\HH_i$ in the solution and $v$ votes for $a$. Since every vertex in $\UU_M$ has exactly one outgoing neighbor, all the incoming flow at every vertex in $\UU_M$ move to their corresponding neighbor in $\UU_R$. Similarly, the outgoing neighbor of every vertex in $\UU_R$ is $t^\pr$, all the incoming flow at every vertex in $\UU_R$ move to $t^\pr$. Obviously the flow conservation property is satisfied at every vertex. Also capacity and demand constraints are also satisfied at every edge since the guess corresponds to a solution of the \FD instance. Finally since the total outgoing flow at $s^\pr$ is $n$, the total flow value is also $n$.
 
 In the other direction, assuming $x_i$ and $y_i$ being a guessed winner and a runner up respectively in $\HH_i$ and $n_i$ being the plurality score of a winner in $\HH_i$ for $i\in[k]$, the corresponding flow network has a flow value of $n$, we claim that the \FD instance is a \YES instance. We can assume without loss of generality that the flow value on every edge in a maximum flow is an integer since the demand and capacity of every edge are integers. We define a voter $v\in\VV$ to be in the district $\HH_i,i\in[k]$ if there exists an alternative $a$ such that there is one unit of flow in the edge $(u_v, u_{a,i})$. It follows from the construction of the maximum flow instance that the above partitioning the voters into the districts $\HH_i$ is valid (that is, it respects $\pi, s_{min},$ and $s_{max}$) and the maximum margin of victory of any district is at most $t$. Hence the \FD instance is also a \YES instance.
\end{proof}

\section{Experimental Evaluation}

\subsection{Greedy Algorithms}
Given the high 
complexities of \FP, \FD, and \FCD problems, we propose a set of fast greedy heuristics to 
minimize the margin of victory by moving voters between districts, while respecting the constraints on mobility of the users, and connectedness. 
We describe the algorithms below, given an initial partition as input.

\begin{itemize}
\item \GP minimizes the maximum margin of victory of all districts by greedily moving voters between districts (starting from voters in the district having highest margin of victory in the initial partition), allowing voters to move to any district.

\item \GD minimizes the maximum margin of victory of all districts by greedily moving voters between districts, where every voter has a constraint on where they can move. 

\item \GCD minimizes the maximum margin of victory of all districts by greedily moving voters between districts such that no district becomes a disconnected subgraph.
\end{itemize}

\subsection{Data}
We collected three main datasets, two real datasets and one synthetic dataset using graph models. The real datasets consist of the general parliament elections in the U.K. from $2017$ and demographic information of students in public schools of Detroit, MI. We evaluate all three greedy algorithms on the synthetic dataset, but only evaluate \GP and \GD on the real dataset as we lack the (social) network information in them.

\begin{figure*}[!t]
\centering
  \subfloat[Real Data]{\includegraphics[width=0.25\textwidth] {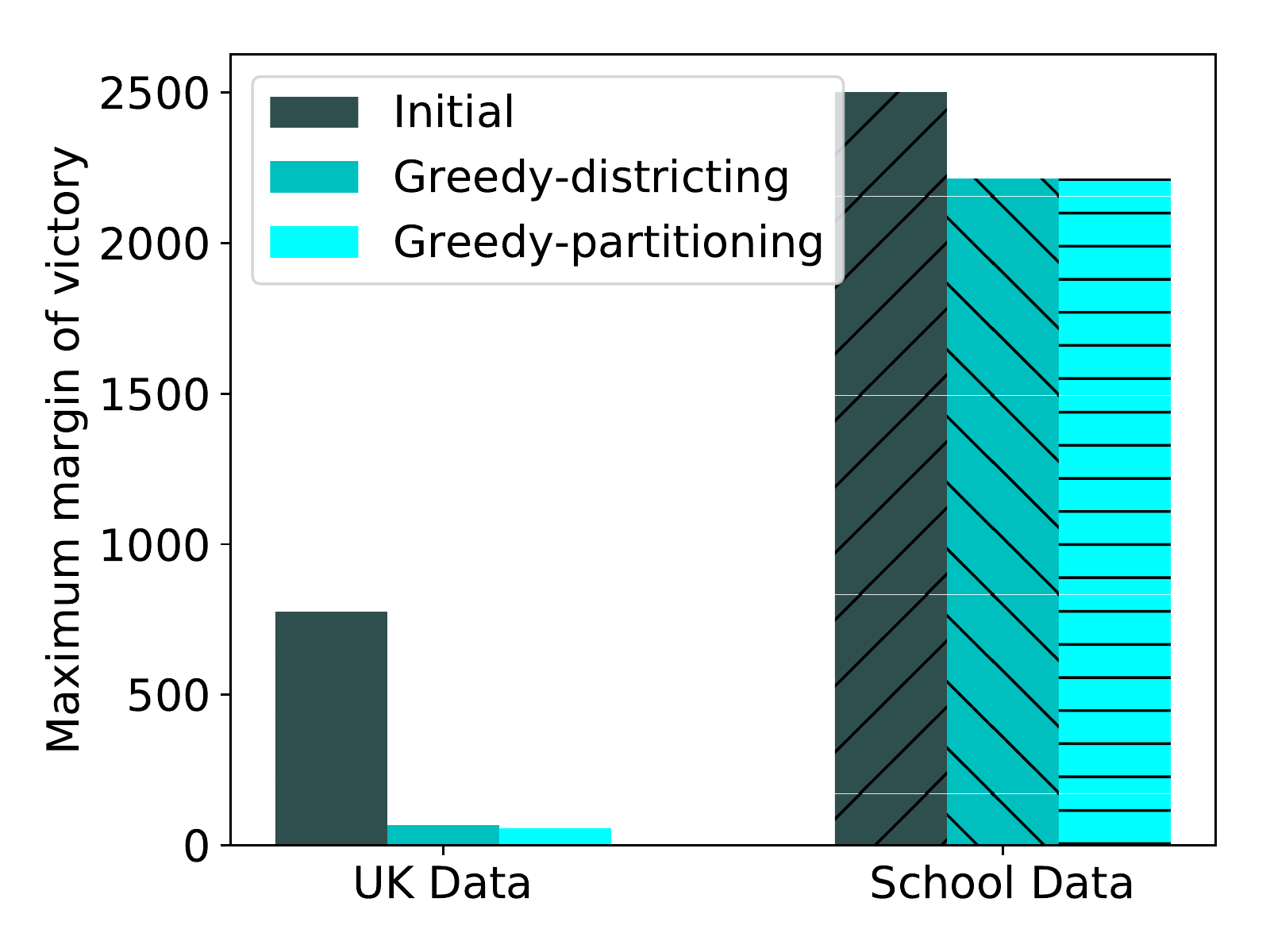}}
  \subfloat[Synthetic Data]{\includegraphics[width=0.25\textwidth] {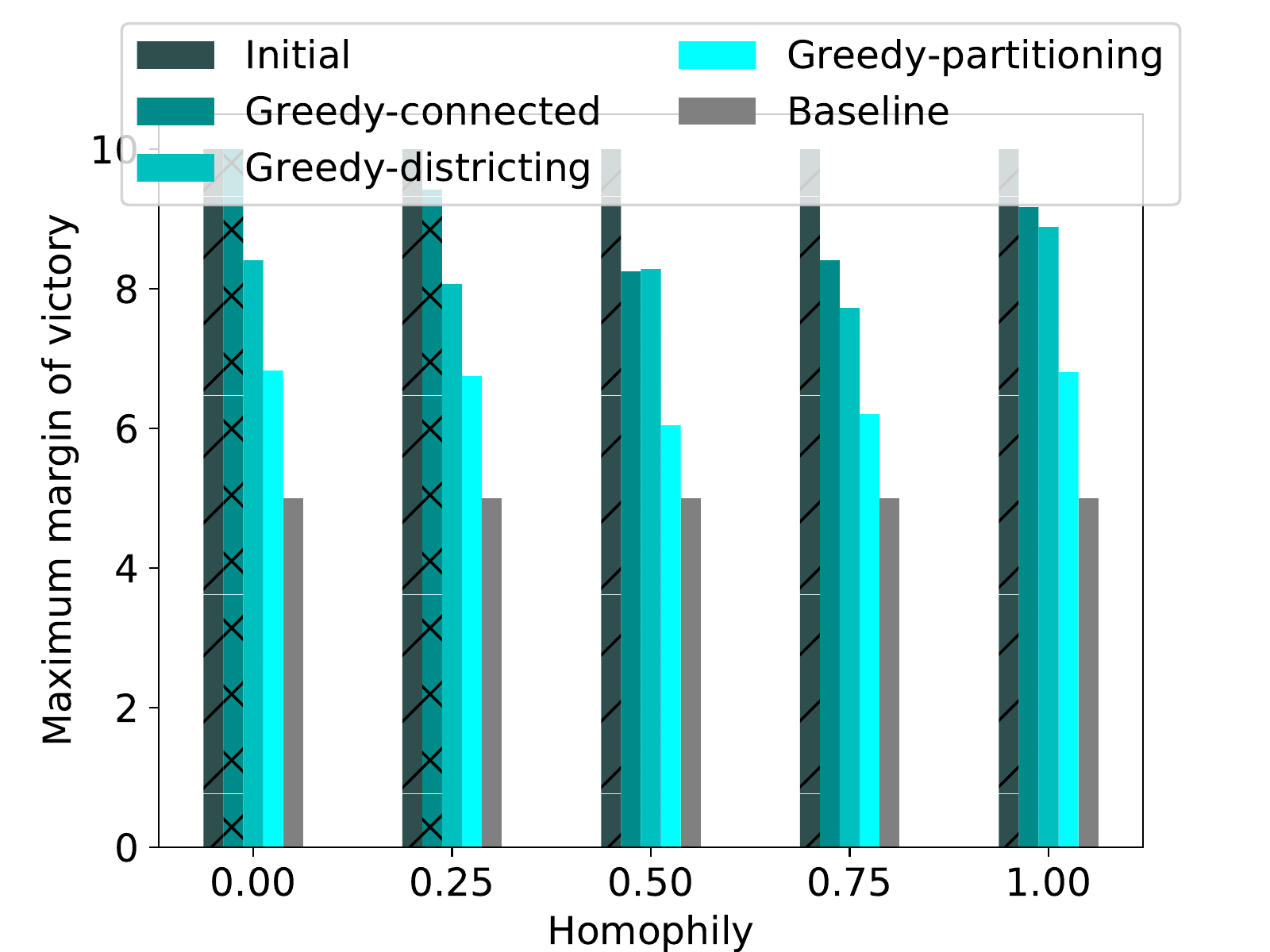}}
  \subfloat[Real Data]{\includegraphics[width=0.25\textwidth] {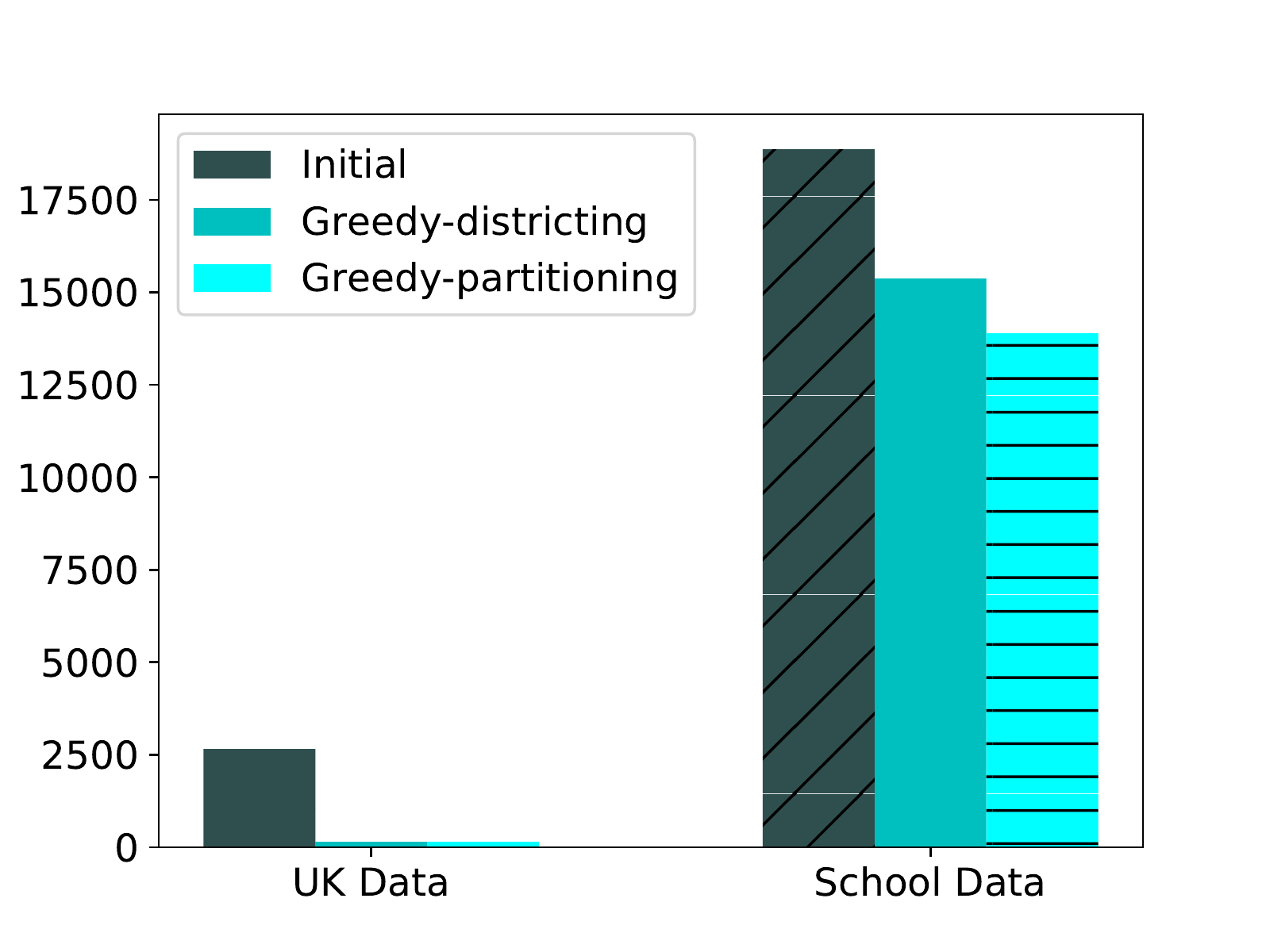}}
  \subfloat[Synthetic Data]{\includegraphics[width=0.25\textwidth] {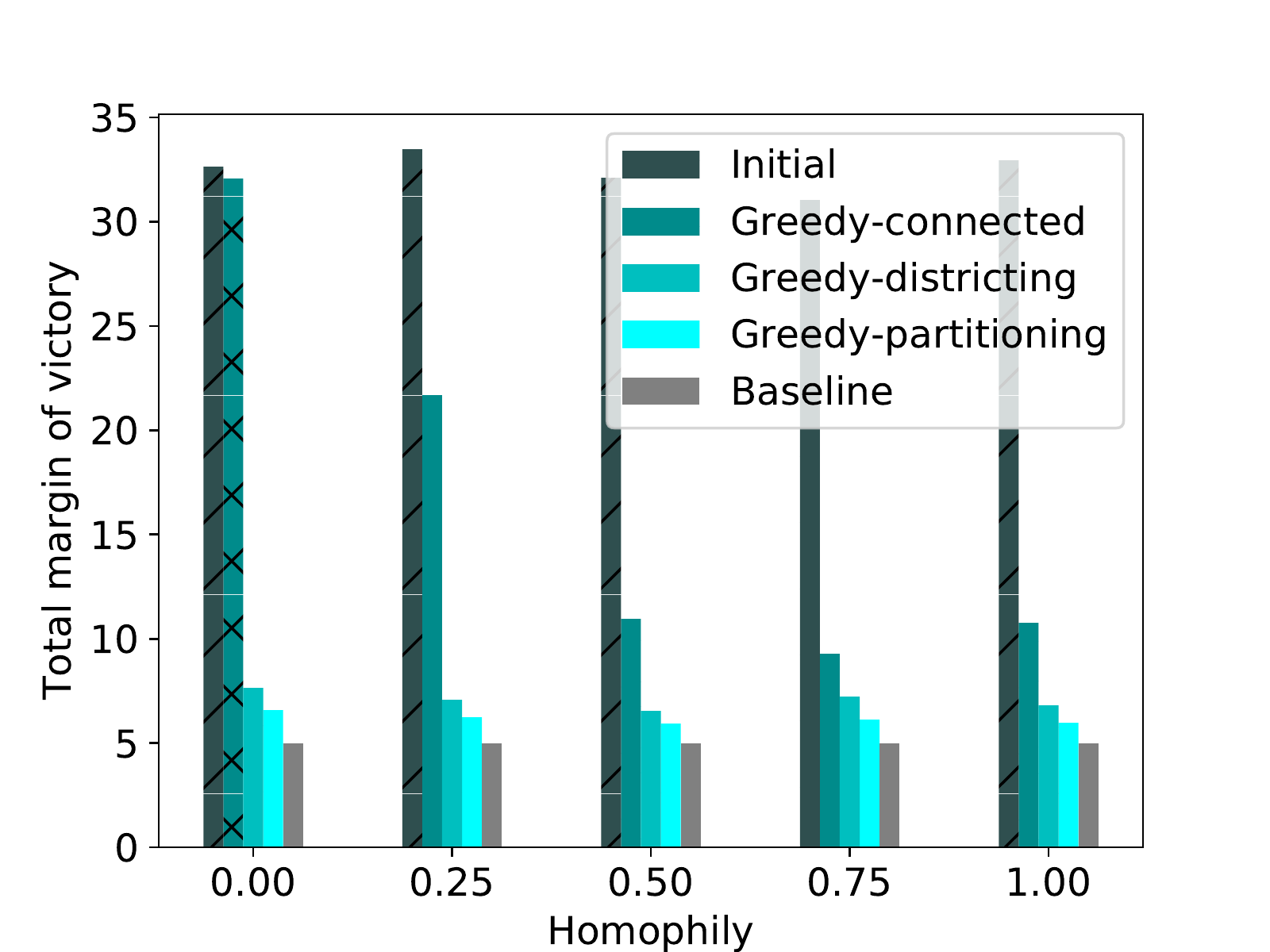}}
  \caption{\bf Maximum margin of victory for all algorithms in (a) real data and (b) synthetic data. Total (sum) of margin of victory in (c) real data and (d) synthetic data.}
  \label{fig:greedya}
  \vspace{-2mm}
\end{figure*}


\paragraph{UK General Elections:} We collected data regarding UK Parliament elections in 2017 from The Electoral Commission ({\tt electoralcommission.org.uk}), 
using constituencies as $districts$ and parties as $alternatives$. Though the votes are cast for individuals, yet in the Parliament number of seats for each party is the number that counts, we are interested in the effect of districting on the distribution of votes over parties rather than over individuals. Knowing the number of votes each party got in each constituency, we simulated the preferences of the voters. 

We tested our algorithm on $10$ neighboring constituencies out of the $650$ in the region of Scotland bordering Edinburgh, which represents a very diverse area in terms of voting preferences. Indeed, these neighboring constituencies voted very differently, each having a clear majority. For example, the distribution of votes in East Lothian was $36.3\%$ Labour, $29.8\%$ Conservative, $30.7$ Scottish National Party (SNP), and $3.1\%$ Liberal Democrats, while in Edinburgh East it was $34.6\%$ Labour, $18.5\%$ Conservative, $42.5\%$ SNP, and $4.2\%$ Liberal Democrats.\footnote{The $10$ constituencies we sampled are: Dumfriesshire, Clydesdale and Tweeddale, Berwickshire, Roxburgh and Selkirk, East Lothian, Midlothian, Edinburgh South, Edinburgh East, Edinburgh North and Leith, Edinburgh South West, Edinburgh West, and Livingston, for which an interactive map with the vote distribution can be found at https://www.bbc.com/news/election-2017-40176349.} We subsampled this dataset, working with a randomized sample of approximately $50,000$ people and we recorded the location of each constituencies (represented by its center), enforcing in \GD that {\it voters can be incentivized to move or to vote only in their closest two constituencies}.

Figure~\ref{fig:greedya}(a) shows that both \GD and \GP are able to reduce the maximum margin of victory of this dataset by approximately $91-92\%$ percent, from $776$ to $67$ and $55$, respectively. Figure~\ref{fig:greedya}(c) shows the effect greedy had on minimizing the overall margin of victory, showing an even larger decrease by almost $95\%$, from $2652$ to $148$ and $135$, respectively. Since \GD represents the more realistic application, we show in Figure~\ref{fig:greedyukdistr} its effect on the voters' distribution in East Lothian and Edinburgh East, showing that it created a stronger opposition for the leading parties (for Labour in East Lothian and for SNP in Edinburgh East). 

\begin{figure}[t]
\centering
\subfloat[East Lothian]{\includegraphics[width=0.25\textwidth] {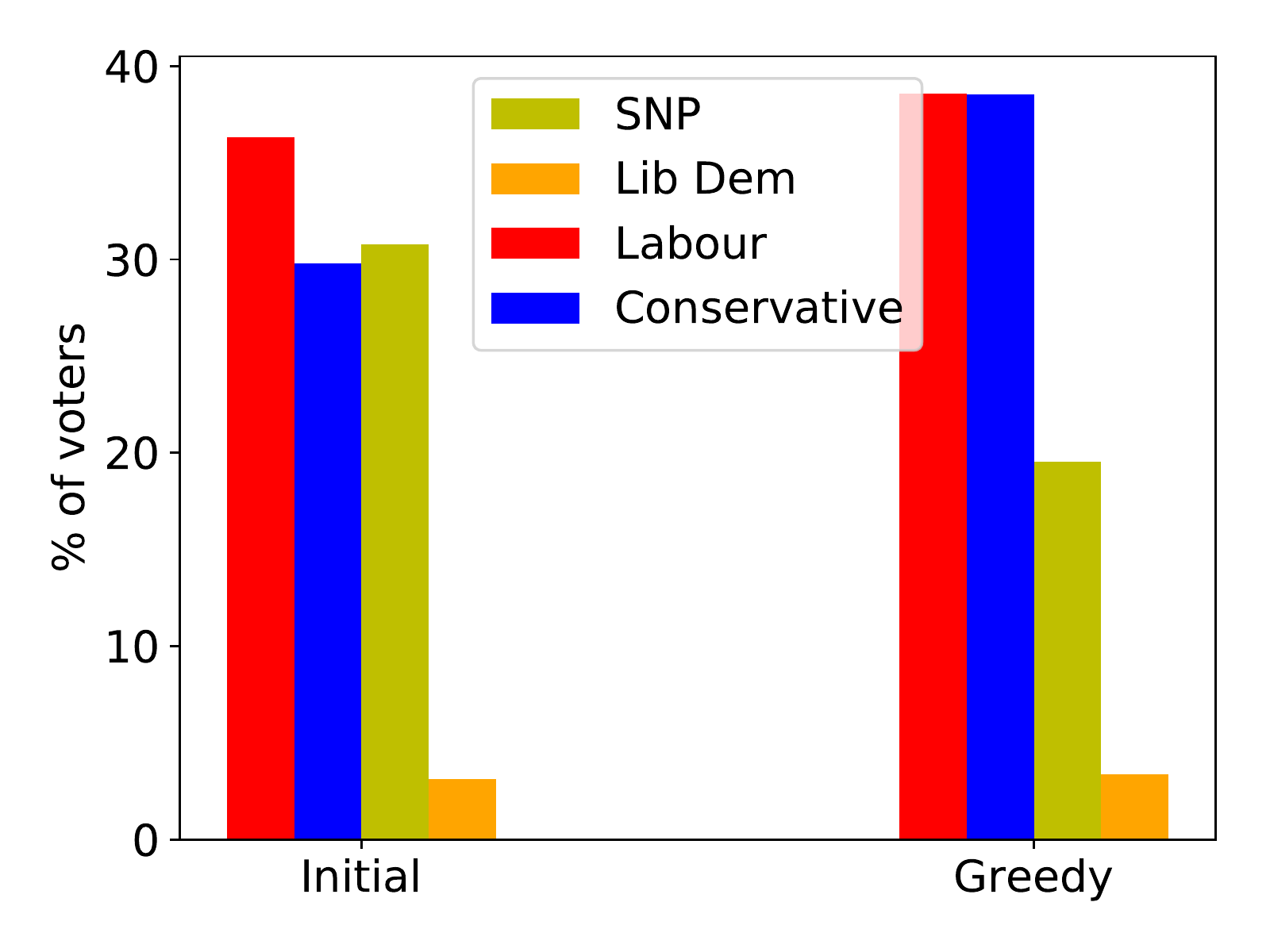}}
\subfloat[Edinburgh East]{\includegraphics[width=0.25\textwidth] {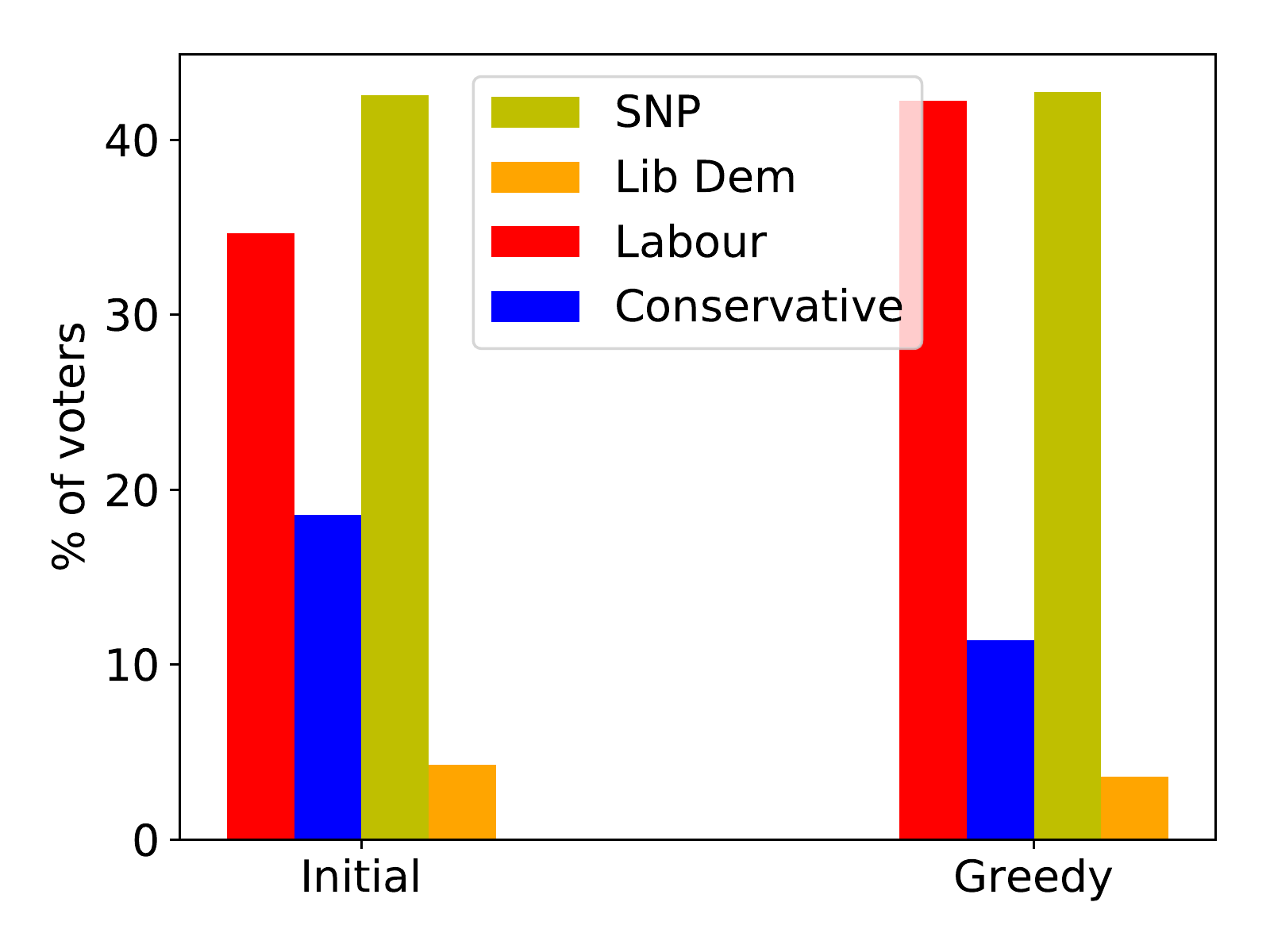}}
\caption{\bf Voters' distribution in UK constituencies before and after applying \GD.}
\vspace{-2mm}
\label{fig:greedyukdistr}
\end{figure}

\begin{figure*}[!htb]
\subfloat[Dove Academy]
{\includegraphics[width=0.33\textwidth] {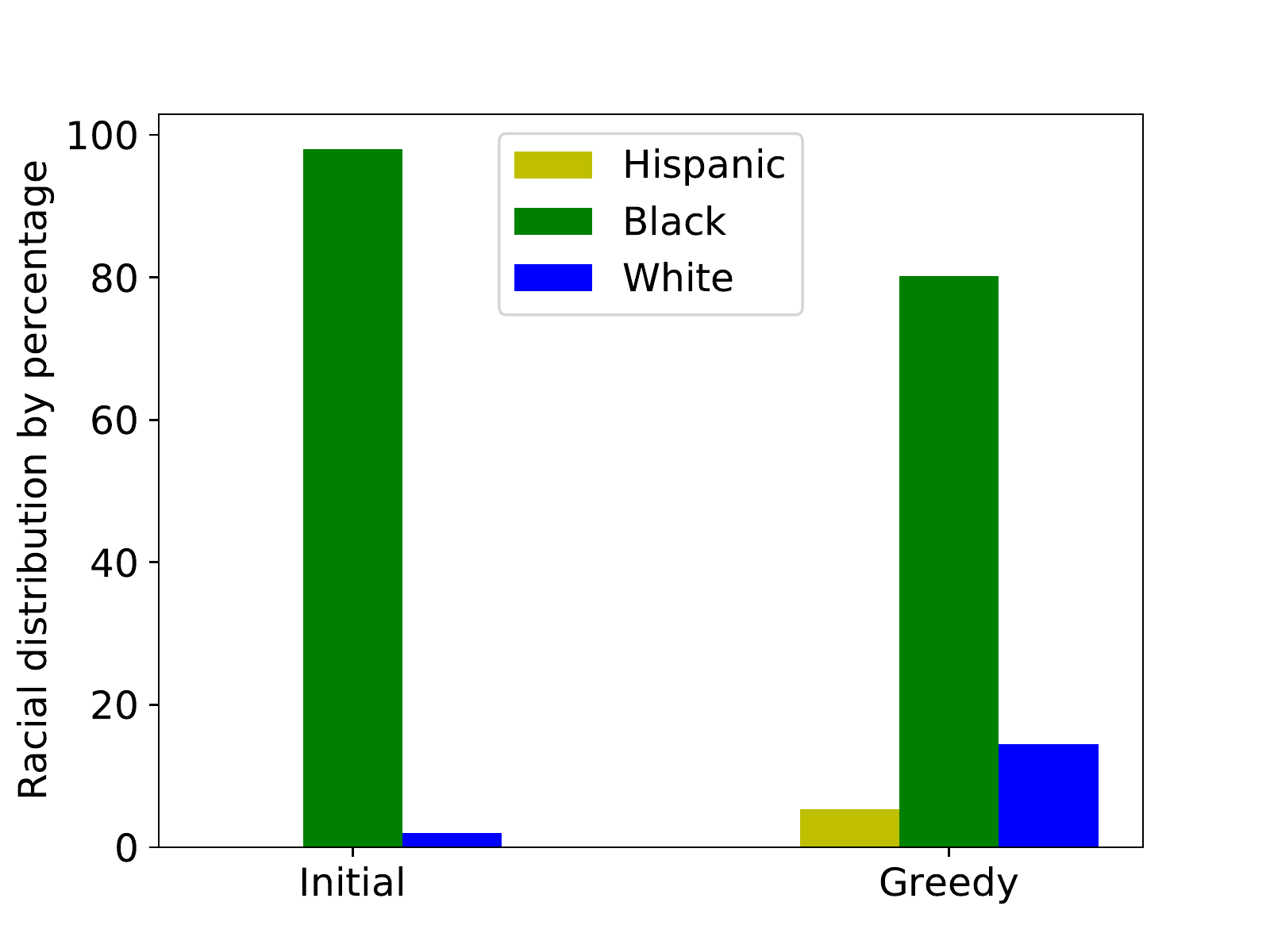}} \subfloat[Universal Academy]
{\includegraphics[width=0.33\textwidth] {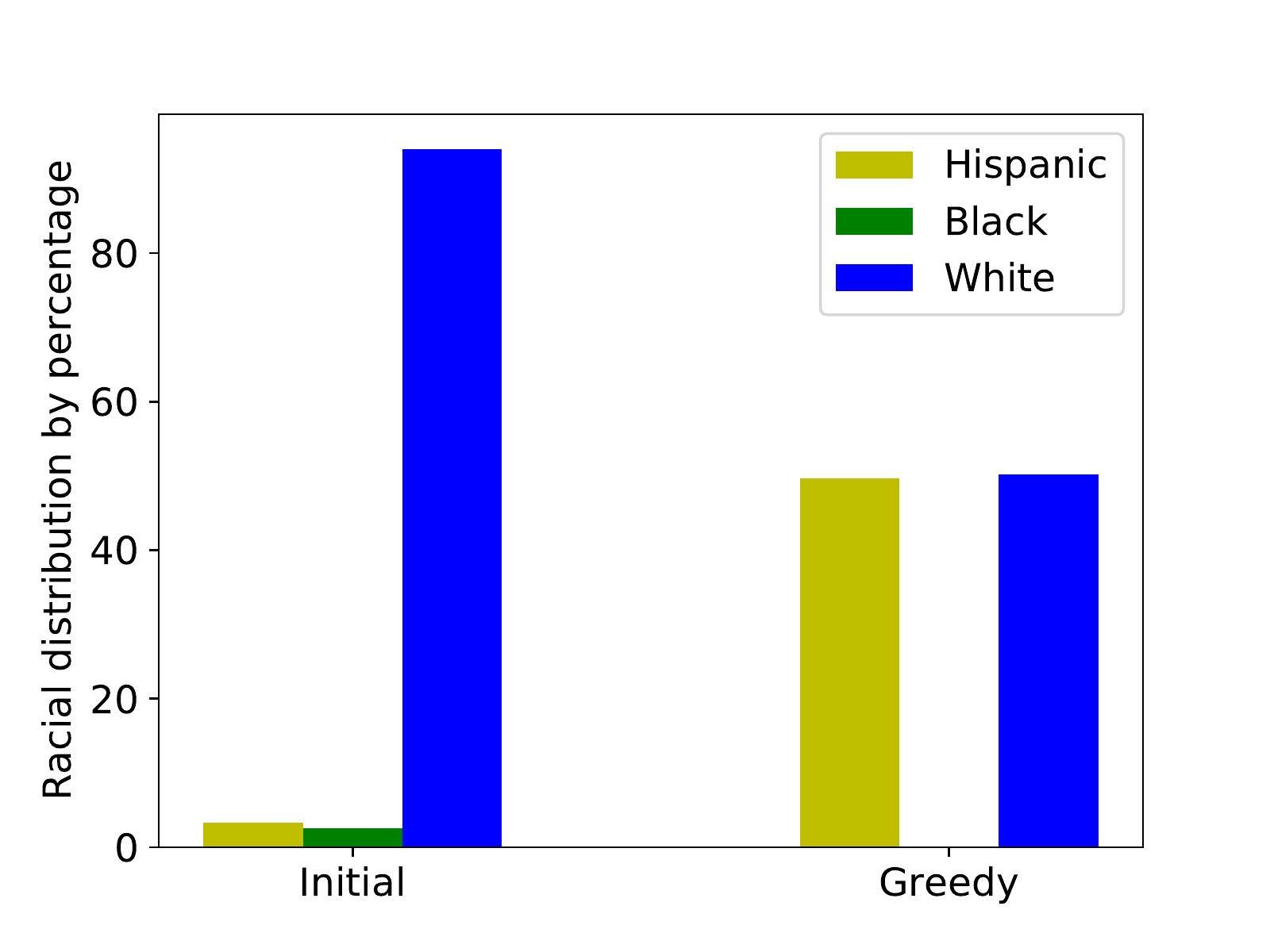}} \subfloat[Cesar Chavez Academy]
{\includegraphics[width=0.33\textwidth] {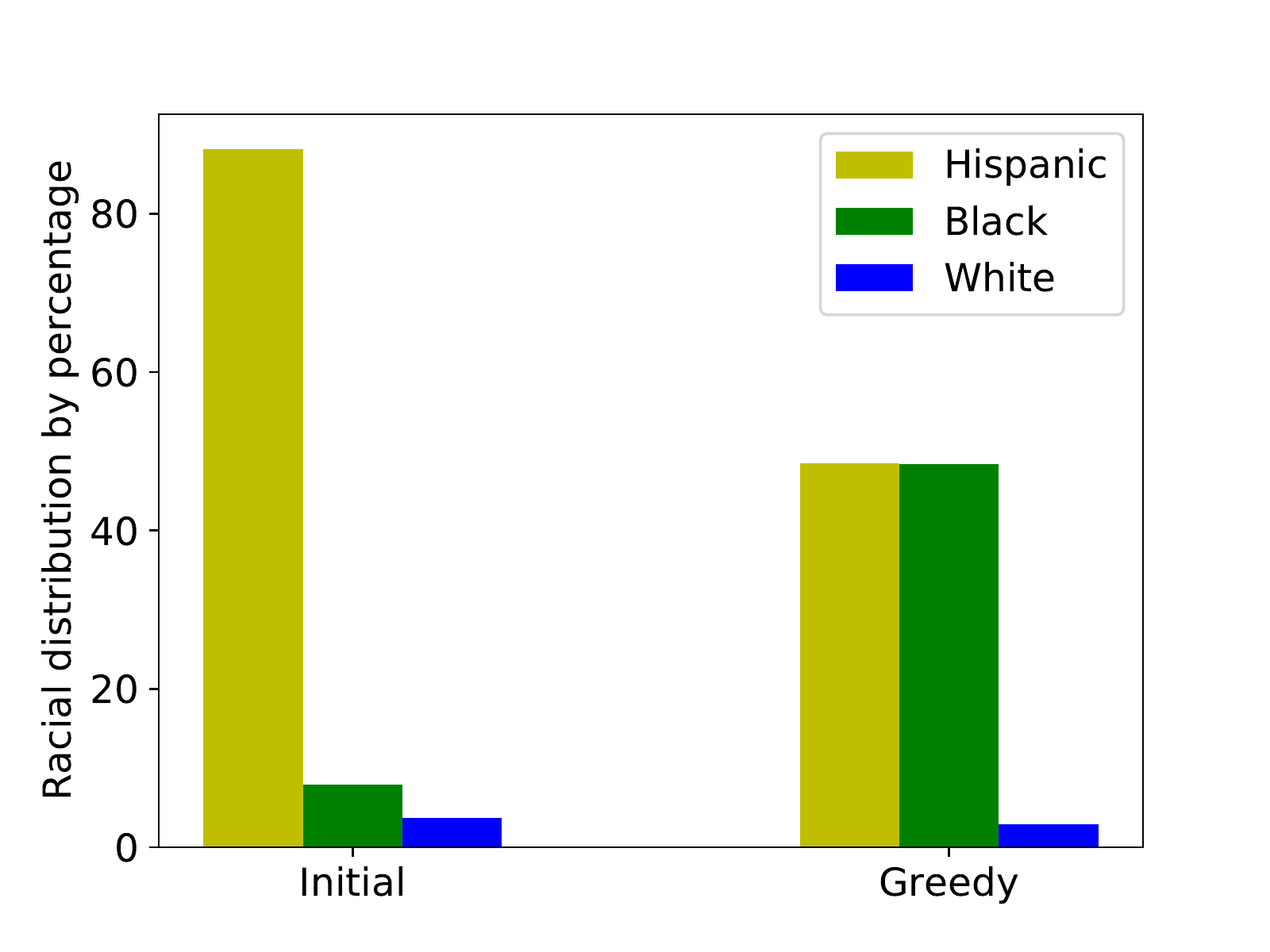}}
\caption{\bf Racial distribution of students in selected schools before and after applying \GD.}
\label{fig:greedyschooldistr}
\vspace*{-2mm}
\end{figure*}

\vspace{-1mm}
\paragraph{US Public School Districting:} 
Neighborhood racial segregation is still widespread in many places in US, 
trickling down to segregation in schools~\cite{franken2019what,richards2014gerrymandering}. One of the main consequences of that is  white-majority schools 
receiving substantially more funding than schools with mostly students of color
~\cite{EdBuild}. 
In this paper, we attempt to show that our algorithms can be used to increase racial diversity (and lower segregation) in schools, if accompanied by government policies that facilitate movement of students between schools~\cite{montgomery1970swann}. 

We collected school data from the National Center for Education Statistics (NCES: {\tt nces.ed.gov/ccd}) about public schools in Detroit, MI, 
which is still one of the cities with highest rates of segregation and most economic and social struggles encountered by minorities~\cite{Detroit1,Detroit2}. We gathered data from $61$ schools in Detroit, each containing between $40$ and $5000$ students, summing up to $41,834$ students and their reported race. We modeled this data in the form of an election, where the $voters$ are the students and the $alternatives$ are their reported race 
(NCES data has $7$ reported races: Asian, Native American, Hispanic, Black, White, Hawaiian, and Mixed-race). Given each student's race, we modeled this `election' as a plurality voting scenario, where each student only `votes' for their reported race. Furthermore, we recorded the location of each school, enforcing in \GD that {\it students can only go to 
their closest $5$ schools.}

Figure~\ref{fig:greedya}(a) shows that both \GD and \GP decrease the total margin of victory by $11-12\%$ on average, from $2,501$ to $2,213$ and $2,311$, respectively. Again, Figure~\ref{fig:greedya}(c) shows the overall decrease in margin of victory by the greedy algorithms, showing a larger decrease of $18-24\%$, from $18,870$ to $15,360$ and $14,376$, respectively. As \GD represents the more realistic scenario, we present in Figure~\ref{fig:greedyschooldistr} its effect on the racial distribution of students in a selection of three schools. we observe that schools containing students from one predominant racial group become more equilibrated: Dove Academy goes from having $98\%$ Black students and $2\%$ White to having $80\%$ Black, $14.5\%$ White, and $5.5\%$ Hispanic students, Universal Academy goes from having $95\%$ White students, $3.5\%$ Hispanic, and $2.5\%$ Black to having $50.15\%$ White and $49.7\%$ Hispanic, while Cesar Chavez Academy goes from having $88\%$ Hispanic students, $8\%$ Black, and $3\%$ White to having $48\%$ Hispanic, $48\%$ Black, and $3\%$ White students.
Given the discrepancy between White-majority schools 
with other schools, we hope that such demographic changes can help in equalizing funding for all students. 

Of course, since minimizing margin of victory only considers the most predominant two races, we may need to enforce an additional diversity constraint to preserve a minimum fraction of students from other races 
in a school (e.g., the $2.5\%$ Black students in Universal Academy may need to stay). 
We leave exploring this direction for future work.

\vspace{-1mm}
\paragraph{Graph Simulations:} To further understand the relationship between margin of victory and population structure, we simulated a set of graphs based on the Erdos-Renyi (ER) graph model, varying the level of connectivity between people with similar political leanings. Unable to vary such a parameter in the real data, we turn to classical graph models to do so. 
Following the methodology in~\cite{cohen2018gerrymandering}, we used the line model to simulate $voters$, $alternatives$, and $voters'$ political affiliation. We then created $50$ instances of the ER graph model, where each node represents a $voter$ and the edges are formed according to the model with an added homophily factor based on the distance between nodes (as simulated by the line model). For every node,  we generated the list of preferences over the candidates 
according to the distance between the voter and the candidates. The inputs to each such graph are the number of voters N (100), the number of candidates C (5), the number of districts K (5), and a homophily parameter.

Such models capture the network and clustering effects exhibited by voter districts in real world~\cite{keegan2016blue,adamic2005political,conover2011political}. We then test all three greedy algorithms in minimizing the maximum margin of victory by re-districting the population in these graphs. We further add a {\bf baseline algorithm} that computes the optimal partition of voters into districts given a network, the districts' size constraints, and mobility of voters. While \GP and \GD come as natural formulations, we argue the need for \GCD as well, since re-districting cannot be done arbitrarily and will be more effective if the population remain connected. 
Thus, in our simulations, the graph models aim to mimic the natural connections individuals make.



We simulated the greedy algorithm for each graph instance, averaging over $10$ iterations the minimal maximum margin of victory that it can reach and compared that to the baseline value. Figure~\ref{fig:greedya}(b) shows the effect of these algorithms in improving the maximum margin of victory aggregated for all graph instances, varying the homophily factor and allowing districts to change up to $20\%$ in size. We observe that no matter how homophilic the initial graph is, greedy is able to successfully reduce the margin of victory for all three algorithms: \GP performs the best as it contains no constraints on mobility of voters, being evaluated close to the baseline value and reducing maximum margin of victory from $10$ to $6-7$ on average, \GD performs second-best, reducing it from $10$ to $8$ on average, while \GCD reduces it from $10$ to $9$ on average, performing worse than the other two due to a tighter connectivity constraint. Figure~\ref{fig:greedya}(d) shows the overall decrease in margin of victory, where the effect is more significant: \GP and \GD achieve a result close to the baseline, while \GCD performs slightly worse, reducing the total margin of victory of $46\%$ on average. The results are qualitatively similar for varying the district size constraints, which we omit due to lack of space.

\section{Conclusion and Future Directions}
In this paper, we tackled the problem of fairly dividing people into groups by conceptualizing the problem in a voting scenario. By modeling the  preferences of people over different candidates, we set the goal to minimize the maximum margin of victory in any group. In doing so, we provide a rigorous framework to reason about the complexity of the problem, 
showing that dividing people with constraints on their neighborhood or their connections is NP-complete in the most general case, and admit  polynomial time algorithms for particular cases. 

Furthermore, we develop and evaluate fast greedy heuristics 
to minimize the maximum margin of victory in practical scenarios. 
Indeed, our results show significant improvement of the margin of victory in the case of elections and school choice, as well as in synthetic experiments. In the case of elections, minimizing margin of victory leads to better representation of opposing parties in electoral districts, where we notice that the opposing parties in the UK can gain more power through re-districting. In the case of school choice, we model students demographic information as an election, where each student 'votes' (or prefers) their own demographic attribute, and show that our greedy algorithms are able to provide more diversity in highly segregated schools. While government policies are ultimately crucial in reducing segregation, we hope that this quantitative analysis can motivate them and show their potential efficacy. 

Multiple directions  remain open for future work. We plan to (i) include an analysis of the social connections in real datasets that may further constrain people's mobility, and (ii) extend synthetic experiments to other graph models as well. Finally, it would be interesting to (iii) measure the effect of minimizing the margin of victory on different gerrymandering metrics, 
and (iv) investigate whether lowering racial segregation would lead more equitable distribution of revenues to public schools.

\balance
\bibliographystyle{aaai}
\bibliography{Main}

\begin{thebibliography}{}

\bibitem[\protect\citeauthoryear{Adamic and Glance}{2005}]{adamic2005political}
Adamic, L.~A., and Glance, N.
\newblock 2005.
\newblock The political blogosphere and the 2004 us election: divided they
  blog.
\newblock In {\em Proceedings of the 3rd international workshop on Link
  discovery},  36--43.
\newblock ACM.

\bibitem[\protect\citeauthoryear{Bachrach \bgroup et al\mbox.\egroup
  }{2016}]{bachrach2016misrepresentation}
Bachrach, Y.; Lev, O.; Lewenberg, Y.; and Zick, Y.
\newblock 2016.
\newblock Misrepresentation in district voting.
\newblock In {\em IJCAI},  81--87.

\bibitem[\protect\citeauthoryear{Barbara and
  Garcia-Molina}{1987}]{barbara1987reliability}
Barbara, D., and Garcia-Molina, H.
\newblock 1987.
\newblock The reliability of voting mechanisms.
\newblock {\em IEEE Transactions on Computers} (10):1197--1208.

\bibitem[\protect\citeauthoryear{Barber{\`a} \bgroup et al\mbox.\egroup
  }{1991}]{barbera1991voting}
Barber{\`a}, S.; Sonnenschein, H.; Zhou, L.; et~al.
\newblock 1991.
\newblock Voting by committees.
\newblock {\em Econometrica} 59(3):595--609.

\bibitem[\protect\citeauthoryear{Berman, Karpinski, and
  Scott}{2003}]{ECCC-TR03-022}
Berman, P.; Karpinski, M.; and Scott, A.~D.
\newblock 2003.
\newblock Approximation hardness and satisfiability of bounded occurrence
  instances of {SAT}.
\newblock {\em Electronic Colloquium on Computational Complexity {(ECCC)}}
  10(022).

\bibitem[\protect\citeauthoryear{Blom, Stuckey, and
  Teague}{2018}]{blom2018computing}
Blom, M.; Stuckey, P.~J.; and Teague, V.~J.
\newblock 2018.
\newblock Computing the margin of victory in preferential parliamentary
  elections.
\newblock In {\em International Joint Conference on Electronic Voting},  1--16.
\newblock Springer.

\bibitem[\protect\citeauthoryear{Borodin \bgroup et al\mbox.\egroup
  }{2018}]{borodin2018big}
Borodin, A.; Lev, O.; Shah, N.; and Strangway, T.
\newblock 2018.
\newblock Big city vs. the great outdoors: Voter distribution and how it
  affects gerrymandering.
\newblock In {\em IJCAI},  98--104.

\bibitem[\protect\citeauthoryear{Butler}{1992}]{butler1992redrawing}
Butler, D.
\newblock 1992.
\newblock The redrawing of parliamentary boundaries in britain.
\newblock {\em British Elections and Parties Yearbook} 2(1):5--12.

\bibitem[\protect\citeauthoryear{Chakraborty \bgroup et al\mbox.\egroup
  }{2019}]{chakraborty2019nudging}
Chakraborty, A.; Mota, N.; Biega, A.~J.; Mohammadi, N.; Gummadi, K.~P.; and
  Heidari, H.
\newblock 2019.
\newblock Nudging toward equitable online donations: A case study of
  educational charity.

\bibitem[\protect\citeauthoryear{Chang}{2018}]{chang2018we}
Chang, A.
\newblock 2018.
\newblock We can draw school zones to make classrooms less segregated. this is
  how well your district does.

\bibitem[\protect\citeauthoryear{Cohen-Zemach, Lewenberg, and
  Rosenschein}{2018}]{cohen2018gerrymandering}
Cohen-Zemach, A.; Lewenberg, Y.; and Rosenschein, J.~S.
\newblock 2018.
\newblock Gerrymandering over graphs.
\newblock In {\em Proceedings of the 17th International Conference on
  Autonomous Agents and MultiAgent Systems},  274--282.
\newblock International Foundation for Autonomous Agents and Multiagent
  Systems.

\bibitem[\protect\citeauthoryear{Conover \bgroup et al\mbox.\egroup
  }{2011}]{conover2011political}
Conover, M.~D.; Ratkiewicz, J.; Francisco, M.; Gon{\c{c}}alves, B.; Menczer,
  F.; and Flammini, A.
\newblock 2011.
\newblock Political polarization on twitter.
\newblock In {\em Fifth international AAAI conference on weblogs and social
  media}.

\bibitem[\protect\citeauthoryear{Corsi-Bunker}{2015}]{corsi2015guide}
Corsi-Bunker, A.
\newblock 2015.
\newblock Guide to the education system in the united states.
\newblock {\em University of Minnesota} 23.

\bibitem[\protect\citeauthoryear{Dey and Narahari}{2015}]{dey2015estimating}
Dey, P., and Narahari, Y.
\newblock 2015.
\newblock Estimating the margin of victory of an election using sampling.
\newblock In {\em Twenty-Fourth International Joint Conference on Artificial
  Intelligence}.

\bibitem[\protect\citeauthoryear{Douglas N.~Harris}{2015}]{urban2015closest}
Douglas N.~Harris, M. F.~L.
\newblock 2015.
\newblock The identification of schooling preferences: methods and evidence
  from post-katrina new orleans.

\bibitem[\protect\citeauthoryear{EdBuild}{2019}]{EdBuild}
EdBuild.
\newblock 2019.
\newblock Non-white school districts get \$23 billion less than white
  districts, despite serving the same number of students.

\bibitem[\protect\citeauthoryear{Erd{\'e}lyi, Hemaspaandra, and
  Hemaspaandra}{2015}]{erdelyi2015more}
Erd{\'e}lyi, G.; Hemaspaandra, E.; and Hemaspaandra, L.~A.
\newblock 2015.
\newblock More natural models of electoral control by partition.
\newblock In {\em International Conference on Algorithmic DecisionTheory},
  396--413.
\newblock Springer.

\bibitem[\protect\citeauthoryear{Feix \bgroup et al\mbox.\egroup
  }{2008}]{feix2008majority}
Feix, M.~R.; Lepelley, D.; Merlin, V.; Rouet, J.-L.; and Vidu, L.
\newblock 2008.
\newblock Majority efficient representation of the citizens in a federal union.
\newblock {\em Manuscript, Universit{\'e} de la R{\'e}union, Universit{\'e} de
  Caen, and Universit{\'e} d’Orl{\'e}ans}.

\bibitem[\protect\citeauthoryear{Felsenthal and
  Miller}{2015}]{felsenthal2015election}
Felsenthal, D.~S., and Miller, N.~R.
\newblock 2015.
\newblock What to do about election inversions under proportional
  representation?
\newblock {\em Representation} 51(2):173--186.

\bibitem[\protect\citeauthoryear{Frankenberg}{2019}]{franken2019what}
Frankenberg, E.
\newblock 2019.
\newblock What school segregation looks like in the us today, in 4 charts.

\bibitem[\protect\citeauthoryear{Gelman, Katz, and
  Tuerlinckx}{2002}]{gelman2002mathematics}
Gelman, A.; Katz, J.~N.; and Tuerlinckx, F.
\newblock 2002.
\newblock The mathematics and statistics of voting power.
\newblock {\em Statistical Science}  420--435.

\bibitem[\protect\citeauthoryear{Institute}{2018}]{Detroit1}
Institute, U.
\newblock 2018.
\newblock Segregated neighborhoods, segregated schools?

\bibitem[\protect\citeauthoryear{Issacharoff}{2002}]{issacharoff2002gerrymandering}
Issacharoff, S.
\newblock 2002.
\newblock Gerrymandering and political cartels.
\newblock {\em Harvard Law Review} 116.

\bibitem[\protect\citeauthoryear{Ito \bgroup et al\mbox.\egroup
  }{2019}]{ito2019algorithms}
Ito, T.; Kamiyama, N.; Kobayashi, Y.; and Okamoto, Y.
\newblock 2019.
\newblock Algorithms for gerrymandering over graphs.
\newblock In {\em Proceedings of the 18th International Conference on
  Autonomous Agents and MultiAgent Systems},  1413--1421.
\newblock International Foundation for Autonomous Agents and Multiagent
  Systems.

\bibitem[\protect\citeauthoryear{Johnston, Rossiter, and
  Pattie}{2006}]{johnston2006disproportionality}
Johnston, R.; Rossiter, D.; and Pattie, C.
\newblock 2006.
\newblock Disproportionality and bias in the results of the 2005 general
  election in great britain: evaluating the electoral system’s impact.
\newblock {\em Journal of Elections, Public Opinion and Parties} 16(1):37--54.

\bibitem[\protect\citeauthoryear{Keegan}{2016}]{keegan2016blue}
Keegan, J.
\newblock 2016.
\newblock Blue feed, red feed.
\newblock {\em The Wall Street Journal} 18.

\bibitem[\protect\citeauthoryear{Kent and Thomas C.~Frohlich}{2015}]{Detroit2}
Kent, A., and Thomas C.~Frohlich, H.~P.
\newblock 2015.
\newblock The 9 most segregated cities in america.

\bibitem[\protect\citeauthoryear{Lewenberg, Lev, and
  Rosenschein}{2017}]{lewenberg2017divide}
Lewenberg, Y.; Lev, O.; and Rosenschein, J.~S.
\newblock 2017.
\newblock Divide and conquer: Using geographic manipulation to win
  district-based elections.
\newblock In {\em AAMAS},  624--632.

\bibitem[\protect\citeauthoryear{Lublin}{1999}]{lublin1999paradox}
Lublin, D.
\newblock 1999.
\newblock {\em The paradox of representation: Racial gerrymandering and
  minority interests in Congress}.
\newblock Princeton University Press.

\bibitem[\protect\citeauthoryear{Montgomery~III}{1970}]{montgomery1970swann}
Montgomery~III, J.
\newblock 1970.
\newblock Swann v. charlotte-mecklenburg board of education: Roadblocks to the
  implementation of brown.
\newblock {\em Wm. \& Mary L. Rev.} 12:838.

\bibitem[\protect\citeauthoryear{Richards}{2014}]{richards2014gerrymandering}
Richards, M.~P.
\newblock 2014.
\newblock The gerrymandering of school attendance zones and the segregation of
  public schools: A geospatial analysis.
\newblock {\em American Educational Research Journal} 51(6):1119--1157.

\bibitem[\protect\citeauthoryear{van~'t Hof, Paulusma, and
  Woeginger}{2009}]{DBLP:journals/tcs/HofPW09}
van~'t Hof, P.; Paulusma, D.; and Woeginger, G.~J.
\newblock 2009.
\newblock Partitioning graphs into connected parts.
\newblock {\em Theor. Comput. Sci.} 410(47-49):4834--4843.

\bibitem[\protect\citeauthoryear{Xia}{2012}]{xia2012computing}
Xia, L.
\newblock 2012.
\newblock Computing the margin of victory for various voting rules.
\newblock In {\em Proceedings of the 13th ACM Conference on Electronic
  Commerce},  982--999.
\newblock ACM.

\end{thebibliography}

\end{document}